\documentclass[showpacs,amsmath,amssymb,twocolumn,pra,superscriptaddress]{revtex4-2}

\usepackage{amsfonts,amsthm,mathrsfs,eufrak,xspace,framed}
\usepackage{color}
\usepackage{algorithm}
\floatname{algorithm}{Protocol}
\usepackage{algorithmic}
\usepackage{subfigure}
\usepackage{amscd}
\usepackage{multirow}
\usepackage{graphicx}
\usepackage{dcolumn}
\usepackage{bm}
\usepackage{bbm}
 \usepackage{enumitem}

\newtheorem{lemma}{Lemma}
\newtheorem{theorem}{Theorem}

\newcommand{\bra}[1]{\mbox{$\left\langle #1 \right|$}}
\newcommand{\ket}[1]{\mbox{$\left| #1 \right\rangle$}}

\newcommand{\indep}{\mathrel{\raisebox{0.05em}{\rotatebox[origin=c]{90}{$\models$}}}}

\newcommand{\be}{\begin{equation}}
\newcommand{\ee}{\end{equation}}
\newcommand{\bea}{\begin{eqnarray}}
\newcommand{\eea}{\end{eqnarray}}

\definecolor{mygreen}{rgb}{0,0.5,0}
\definecolor{myblue}{rgb}{0,0,0.75}
\definecolor{mymagenta}{cmyk}{0,1,0,0.12}

\usepackage{umoline}

\begin{document}
\title{Sequential device-independent certification of indefinite causal order}
\author{Zhu Cao}
\email{caozhu@ecust.edu.cn}
\affiliation{Key Laboratory of Smart Manufacturing in Energy Chemical Process, Ministry of Education, East China University of Science and Technology, Shanghai 200237, China}

\begin{abstract}
Indefinite causal order has found numerous applications in quantum computation, quantum communication, and quantum metrology.
Before its usage, the quality of the indefinite causal order needs to be first certified, and the certification should ideally be device-independent (DI) to avoid the impact of device imperfections. In this work, we initiate the study of the sequential DI certification of an indefinite causal order. 
This can be useful in experimental platforms where the generation of an indefinite causal order is difficult.
We show that an arbitrary number of sequential DI certifications of an indefinite causal order can be achieved with a quantum switch and
also analyze practical requirements for experimental implementations of the certifications. Our work opens 
the possibility of reusing the resource of an indefinite causal order multiple times in device-independent quantum information processing.
\end{abstract}

\maketitle

\section{Introduction}

Indefinite causal order, which says there is no definite causal order between several events, is a unique quantum phenomenon that has attracted a lot of interest in recent years \cite{oreshkov2012quantum,chiribella2013quantum,brukner2014quantum,araujo2015witnessing,oreshkov2016causal,barrett2021cyclic}.
The peculiar property of indefinite causal order has found many applications, including
lowering query complexity  \cite{colnaghi2012quantum,araujo2014computational,renner2022computational}, improving channel discrimination accuracy  \cite{chiribella2012perfect},  increasing quantum channel capacity  \cite{ebler2018enhanced,procopio2019communication,goswami2020increasing,caleffi2020quantum,bhattacharya2021random,chiribella2021indefinite,sazim2021classical,chiribella2021quantum}, reducing communication complexity  \cite{guerin2016exponential}, improving the efficiency of thermodynamics  \cite{felce2020quantum,guha2020thermodynamic,simonov2022work}, and enhancing the precision of metrology  \cite{zhao2020quantum,chapeau2021noisy}.
Currently, the quantum switch \cite{chiribella2013quantum} is the only example of indefinite causal order that has been experimentally demonstrated  \cite{procopio2015experimental,rubino2017experimental,goswami2018indefinite,guo2020experimental,rubino2022experimental}. 
The quantum switch, like a classical switch, controls the wirings of a circuit. However, in contrast to a classical switch, the quantum switch can be 
in a superposition of the on and off states.

Before applying an indefinite causal order to applications, it is necessary to first 
certify the existence and also the quality of an indefinite causal order.
Most certification of an indefinite causal order is device-dependent \cite{araujo2015witnessing,bavaresco2019semi,zych2019bell,dourdent2022semi}, where there are certain assumptions
on the device. In reality, these assumptions may or may not hold, putting doubt on the validity of 
the certification. Hence, it is desirable to certify the indefinite causal order in a device-independent (DI) way. 
By device independence, we mean that the certification can be achieved solely from the inputs 
and outputs of the devices, but not any assumptions on the inner workings of the devices, which is similar to the certification of nonlocality with Bell tests \cite{brunner2014bell}.
Recently, the first device-independent test of the indefinite causal order of a quantum switch was developed \cite{van2022device}. It is based 
on the violation of a local-causal inequality. The difference between a local-causal inequality and a normal causal inequality
is that there are two spacelike separated parties, called Charlie and Bob, involved in the local-causal inequality.

The objective of our research is to explore the fundamental limitations of the indefinite causal order, 
with a focus on whether a long series of indefinite causal orders can be sequentially and device-independently certified from a single quantum switch. 
This approach may be practically valuable in experimental platforms with inherent challenges in generating quantum switches, as observed in nitrogen-vacancy centers \cite{hensen2015loophole}. The primary difficulty arises from the fact that a quantum switch necessitates the utilization of one subsystem from a maximally entangled state as a control qubit. However, in the case of nitrogen-vacancy centers, the generation of a single maximally entangled state is a time-consuming process, limited to once per hour \cite{hensen2015loophole}. More precisely, we consider the certification of the indefinite causal order in a quantum switch by Charlie and $k$ 
independent Bobs. Charlie performs the certification with one of the Bobs sequentially. Here, the Bobs are 
independent in the sense that they do not share their measurement settings and outcomes.

We show that two sequential certifications of the indefinite causal order in a quantum switch are possible with an explicit instance of the certifications. 
We also give the maximum value of the violation for the case of two sequential tests. The margin is substantial, where the classical bound is $7/4$ and the quantum value for both rounds is more than 1.764. 
Next, we check loopholes for two sequential tests, including the detection loophole and 
randomness loophole. We show that there is 
no detection loophole when the detection efficiency is larger than $99.6\%$. For the randomness loophole, we show that when the min-entropy of the randomness 
is at least 1.92,  there is no randomness loophole. We also examine the case of $k>2$ sequential
 certifications of an indefinite causal order. We show that such certifications are possible for arbitrarily large $k$.
This in particular implies that one can indeed achieve a long series of sequential device-independent certifications of an indefinite causal order with a single quantum switch. For proving this result, we have given an explicit construction
for $k$ sequential tests, where $k$ is any positive integer.
Our work opens the possibility of reusing the resource of an indefinite causal order multiple times in device-independent quantum information processing.

\section{Sequential local-causal inequality}The sequential certification scenario is an extension of the single certification scenario \cite{van2022device}. In the single certification scenario, there are four parties, Alice 1, Alice 2, Bob, and Charlie. Alice 1, Alice 2, Bob, and Charlie are
given the inputs $x_1$, $x_2$, $y$, and $z$ (called settings) respectively and outputs $a_1$, $a_2$, $b$, and $c$, respectively.
Here, Alice 1 is either causally before Alice 2 or causally after Alice 2. Charlie is always causally after Alice 1 and Alice 2. 
Bob is space separated from Alice 1, Alice 2, and Charlie. The probability distribution to be examined is 
$p( a_1a_2 b c | x_1 x_2 y z)$ in the single certification scenario. There are certain restrictions on this 
probability distribution based on the causal orders between the four parties. 

For the sequential certification scenario, there are $n+3$ parties: Alice 1, Alice 2, Charlie, Bob 1, Bob 2, $\cdots$, Bob $n$, which are given the inputs $x_1$, $x_2$,  $z$,  $y_1$, $\cdots$, 
 $y_n$, respectively, and outputs $a_1$, $a_2$, $c$,  $b_1$, $\cdots$, $b_n$, respectively.
When a definite causal order exists, Alice 1 is either causally before Alice 2 or causally after Alice 2. Charlie is always causally after Alice 1 and Alice 2. All Bob $k$ 
($1\le k \le n$) are spacelike separated from Alice 1, Alice 2, and Charlie. Bob $k+1$ is causally after Bob $k$ ($1\le k \le  n-1$).
An illustration of the causal relations between the parties is shown in Fig.~\ref{fig:relation}(a).
These causal relations put constraints on the probability distribution $p( a_1a_2 b_k c | x_1 x_2 y_k z)$ for $1\le k\le n$.
which we formulate in mathematical languages in the following.

For simplicity, below we use $A_1$ to denote Alice 1, $A_2$ to denote Alice 2, $C$ to denote Charlie, and 
$B_k$ to denote Bob $k$. Also for simplicity, we use  $a \indep_p  y$ to denote the condition 
\begin{equation}
p(a | xy) = p(a |xy'), \quad \forall a,x,y,y'.
\end{equation}
We assume there is a hidden variable $\lambda$ that controls the causal order between 
Alice 1 and Alice 2. 
Let $\lambda=1$ denote Alice 1 is before Alice 2 and $\lambda=2$ denote Alice 2 is before Alice 1.
In other words,
\begin{equation}
\label{eq:ass4}
a_1 \indep_{p^1} x_2, a_2 \indep_{p^2} x_1, 
\end{equation}
where $p^1$ ($p^2$) denotes the probability distribution  $p( a_1a_2 b_k c | x_1 x_2 y_k z)$ in the case of $\lambda=1$ ($\lambda=2$).

The hidden variable is chosen before any party chooses its setting, namely we have
\begin{equation}
\label{eq:ass2}
\lambda \indep_p x_1x_2y_kz, \forall k.
\end{equation}
The condition that all Bob $k$ 
($1\le k \le n$) are space separated from Alice 1, Alice 2, and Charlie implies
\begin{equation}
\label{eq:ass31}
 a_1a_2 c  \indep_{p^\lambda} y_k,   \quad \forall \lambda,k,
\end{equation}
and
\begin{equation}
\label{eq:ass32}
 b_k  \indep_{p^\lambda} x_1x_2z,     \quad \forall \lambda,k.
\end{equation}
The condition that Charlie is causally after Alice 1 and Alice 2 implies
\begin{equation}
\label{eq:ass33}
a_1a_2  \indep_{p^\lambda} z,   \quad \forall \lambda,k.
\end{equation}

Note here that by two parties, say $A_1$ and $C$, being causally ordered, we always mean that 
the output of the previous party occurs before the latter party is given its input. Hence, two parties 
may not be causally ordered in the following example: the input and output of $A_1$ and $A_2$ 
are ordered in time as $x_1 < x_2  < a_1  < a_2$. Hereafter, we assume all inputs and outputs of 
the same party happen at the same time (or happen within a very short time interval) and the events of different parties happen at different 
times. In this case, all parties can be ordered causally in the classical case, which justifies the assumption
that Alice 1 is either causally before Alice 2 or causally after Alice 2.

We call the probability distributions $p$ that satisfy the causal constraints in Eqs.~\eqref{eq:ass4}-\eqref{eq:ass33} ``local-causal''
correlations, and denote the set of these distributions by $\mathcal{LC}$.

\begin{figure}[htb]
\centering \includegraphics[width=8.5cm]{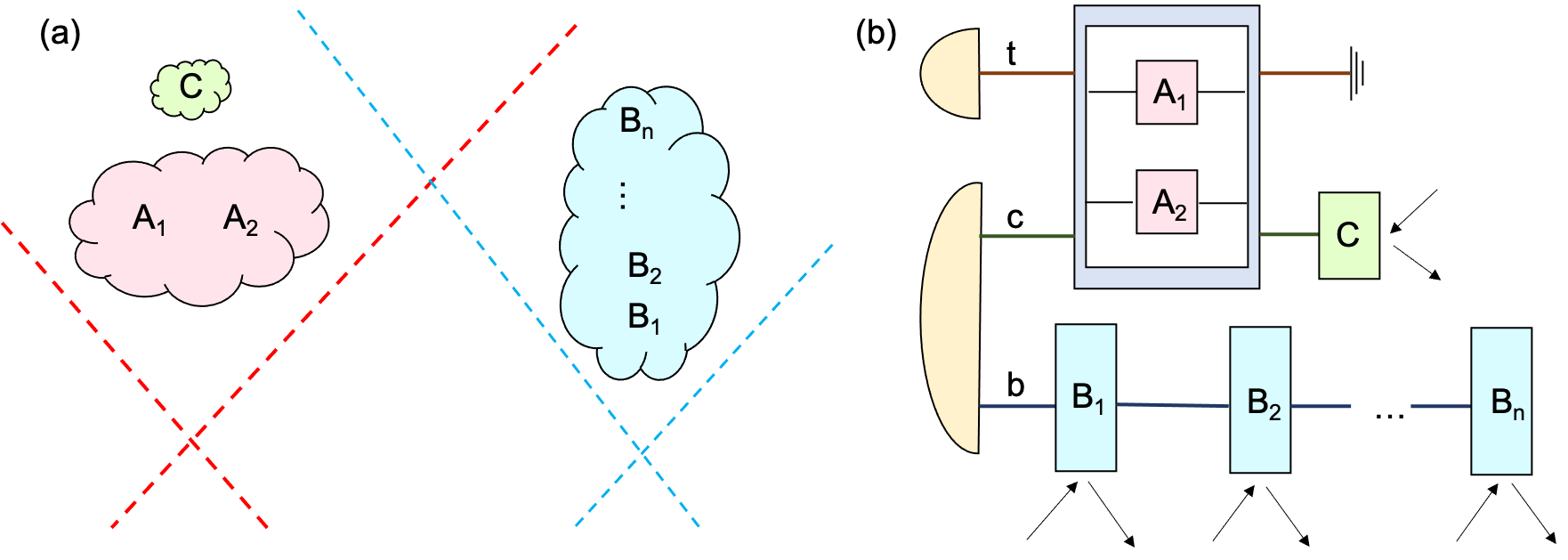}
\caption{(a) The causal relations between the parties. (b) The setup of the sequential certification. }
\label{fig:relation}
\end{figure}


To describe the local-causal (LC) bound, we first define the following quantity for the $k$th round:
\begin{equation}
\begin{aligned}
I^k  =& p( b_k=0, a_2=x_1 | y_k=0 ) + p(  b_k=1 , a_1 =  \\
 &x_2  | y_k=0 ) + p( b_k \oplus c  = y_k z ),
\end{aligned}
\end{equation}
where $a_1,a_2, b_k, c \in \{0,1\}$, $\oplus$ is addition modulo 2, and $x_1, x_2, y_k, z$ 
are independently and uniformly chosen from $\{0,1\}$. When $p\in \mathcal{LC}$, write 
$I^k$ as $I^k_{LC}$.
We have the following local-causal bound.
\begin{theorem}
We have $I _{LC} ^k \le 7/4$ for all $1\le k\le n$.
\end{theorem}
\begin{proof}
Note that the conditions of Theorem 1 in Ref.~\cite{van2022device} are satisfied in any round $k$. 
Hence by applying Theorem 1 in Ref.~\cite{van2022device} to round $k$, we obtain that $I _{LC} ^k \le 7/4$ for any $k$.
\end{proof}

\section{Review of the quantum switch}
As a prerequisite for the violation of the LC bound, we review what is a quantum switch \cite{chiribella2013quantum}. Its setting is as follows. Initially, we have a quantum state
$t$ and we wish to apply two gates $A_1$ and $A_2$ sequentially on this quantum state.
There are two obvious ways to do this. One is to apply $A_1$ before $A_2$, as shown in 
Fig.~\ref{fig:QS}(a). The other is to apply $A_2$ before $A_1$, as shown in Fig.~\ref{fig:QS}(b).
For both of these two cases, there is a definite causal order between the two quantum processes
$A_1$ and $A_2$. The idea of the quantum switch is to add a control qubit, denoted as $c$, to
control the order between $A_1$ and $A_2$. A control qubit $\ket{0}$ denotes $A_1$ is before 
$A_2$. A control qubit $\ket{1}$ denotes $A_2$ is before 
$A_1$. A control qubit $(\ket{1}+\ket{0})/\sqrt{2}$ then represents an indefinite causal order 
between $A_1$ and $A_2$, as illustrated in Fig.~\ref{fig:QS}(c). We will use
an alternative drawing of the quantum switch as depicted in Fig.~\ref{fig:QS}(d) for visual 
clarity. The inner working of Fig.~\ref{fig:QS}(d) is exactly the same as that of Fig.~\ref{fig:QS}(c).

\begin{figure}[htb]
\centering \includegraphics[width=8.5cm]{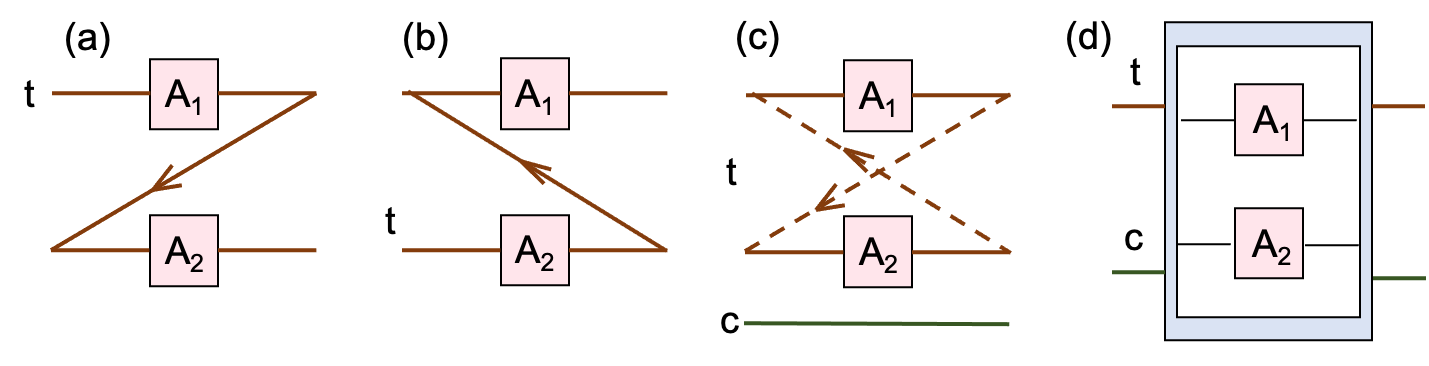}
\caption{A quantum switch. (a) $A_1$ is causally before $A_2$. (b) $A_2$ is causally after $A_1$. (c) The inner working of a quantum switch that acts on $A_1$ and $A_2$. (d) A pictorial representation of a quantum switch. }
\label{fig:QS}
\end{figure}

\section{Setup for sequential violations}After defining the LC bound and reviewing the quantum switch, we are ready to define the setting where we use a quantum switch 
to enable sequential device-independent certification of an indefinite causal order, as illustrated 
in Fig.~\ref{fig:relation}(b).  First, a quantum state is prepared for the control qubit (denoted as $c$ and lying in the Hilbert space $\mathcal{H}_C $),
the target qubit (denoted as $t$ and lying in the Hilbert space $\mathcal{H}_T$), and the input qubit of Bobs' (denoted as $b$).  The control qubit and the target qubit 
are fed into a quantum switch \cite{chiribella2013quantum}, which also takes two quantum channels $A_1$ and $A_2$ as inputs 
and achieves an indefinite causal order for a suitable choice of the control qubit. 
More precisely, $A_1$ and $A_2$ are both quantum channels on the Hilbert space $\mathcal{H}_T$.
The quantum switch $\mathbb{S}$ takes the two quantum channels $A_1=E (\cdot) E^\dagger$ and $A_2=F (\cdot) F^\dagger$ as inputs, where $E$ and $F$ are Kraus operators,  and outputs
a quantum channel $\mathbb{S}(A_1,A_2) =  W (\cdot) W^\dagger$ acting on the Hilbert space $\mathcal{H}_T \otimes \mathcal{H}_C $, where $W =  \ket{0}\bra{0}_c \otimes FE + \ket{1}\bra{1}_c \otimes EF$. Here, $ \ket{0}\bra{0}_c$ and $\ket{1}\bra{1}_c$  act on the Hilbert space $\mathcal{H}_C $.
The quantum output of $\mathbb{S}(A_1,A_2)$ is fed to Charlie, who then takes a classical
input $z$ and gives a classical output $c$.

Bob 1 takes his input quantum state, measures it according to his input setting $y_1$,
and obtains the classical outcome $b_1$. The postmeasurement state is given 
to Bob 2 who then repeats the process of Bob 1. Bob 2 then forwards his output to Bob 3,
Bob 3 to Bob 4, etc., and finally from Bob $n-1$ to Bob $n$.
Crucially Bob $k$ does not reveal his setting and classical outcome to Bob $k+1$. It is in this 
sense that different Bobs are independent. When the probability distribution $p$ is generated 
by this setup, we write $p\in \mathcal{Q}$. When $p\in \mathcal{Q}$, we write $I^k$ as $I^k_Q$.

\section{Existence of the violation of two sequential tests}Before explaining the quantum strategy, let us first describe the postmeasurement states of Bob $k$ 
in more detail. Let us denote the initial state that Charlie and Bob 1 get as $\rho_{CB}^1 $.  
Suppose Bob 1 gets the measurement setting $y$ and obtains the outcome $b$. 
According to the L\"{u}der rule \cite{brown2020arbitrarily}, the postmeasurement state of Bob 1 is 
\begin{equation}
 \rho_{CB}^1 \to  \left( \mathbbm{I} \otimes \sqrt{F^1_{b|y}}  \right) \rho_{CB}^1 \left( \mathbbm{I} \otimes \sqrt{F^1_{b|y}}  \right), 
\end{equation}
where $F^1_{b|y}$ is the positive operator-valued measure (POVM) element corresponding to the measurement setting $y$ and the outcome $b$.
Since we assume Bob 2 is ignorant of Bob 1's measurement setting and outcome, the postmeasurement state of Bob 2 is an
average over all possible measurement settings and outcomes, which has the form
\begin{equation}
\label{eq:poststate}
\rho_{CB}^2 =\frac{1}{2} \sum\limits_{b_1,y_1} \left( \mathbbm{I} \otimes \sqrt{F^1_{b_1|y_1}}  \right) \rho_{CB}^1 \left( \mathbbm{I} \otimes \sqrt{F^1_{b_1|y_1}}  \right).
\end{equation}
Here $ b_1, y_1 \in \{0,1\}$. The coefficient $1/2$ is due to $y_1$ having two values. If $y_1$ has, e.g., three values, the coefficient will be $1/3$.

Next, we give a quantum strategy that violates the LC bound for the case of two rounds, which is illustrated in Fig.~\ref{fig:quantum}.
Let us start from the target qubit $T$. The target qubit is initialized as $\ket{0}_T$. 
The channel of Alice $i$ ($i=1,2$) is a measure-and-prepare one.
After getting the input $x_i$, Alice $i$ measures the target qubit in the computational basis and obtains $\ket{a_i}$.
She then prepares the output quantum state to be $\ket{x_i}$ for subsequent processing.
The order of Alice 1 and Alice 2 is controlled by the control qubit $C$.
After Alice 1 and Alice 2 operate on the target qubit, the target qubit is discarded.

\begin{figure}[htb]
\centering \includegraphics[width=8.5cm]{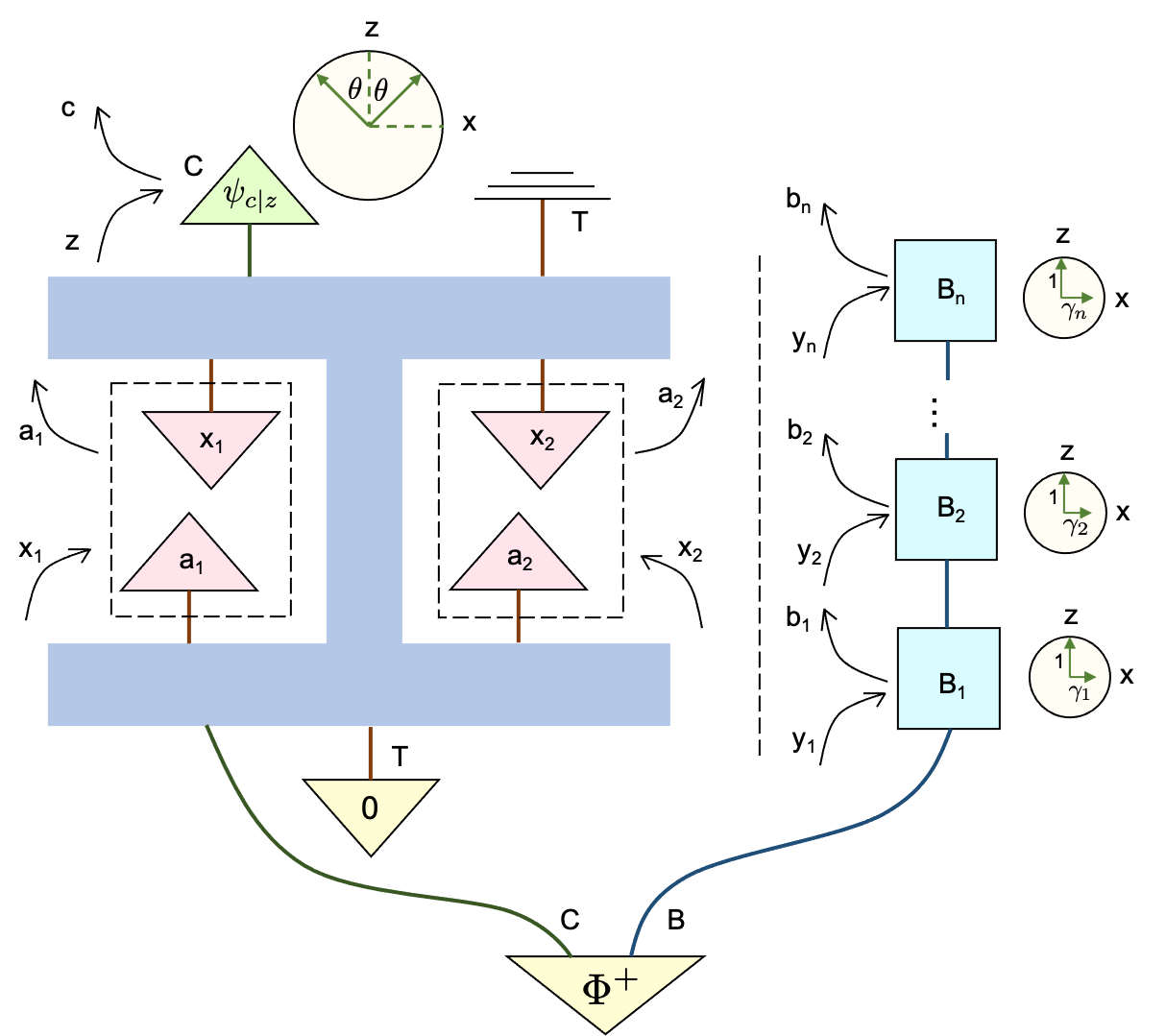}
\caption{The quantum strategy for sequential violations of the L-C bound.}
\label{fig:quantum}
\end{figure}

We now give the measurement strategy for Charlie and Bob. 
The initial quantum state between Charlie and Bob is $\ket{00}+\ket{11}$ where the normalizing factor is omitted for simplicity. 
Denote $C_{c|z}$ as the POVM element of Charlie when he receives the measurement setting $z$ and outputs the outcome $c$. 
Denote $B_{b|y}^k$ as the POVM element of Bob $k$ when he receives the measurement setting $y$ and outputs the outcome $b$. 
We then set the measurement strategy on Charlie and Bob $k$ as follows:
\begin{equation}
\begin{aligned} 
C_{0|0} & = \frac{1}{2} (I + \cos \theta \sigma_z + \sin \theta \sigma_x),  \\ 
C_{0|1} & = \frac{1}{2} (I + \cos \theta \sigma_z - \sin \theta \sigma_x),  \\ 
B_{0|0}^k & = \frac{1}{2} (I + \sigma_z ),  \\ 
B_{0|1}^k & = \frac{1}{2} (I + \gamma_k \sigma_x). 
\end{aligned}
\label{eq:measurement}
\end{equation}
Here, $I$ is a $2\times 2$ identity matrix, and $\sigma_x$ and $\sigma_z$ are Pauli $X$ and $Z$ matrices, respectively.
There are $n+1$ parameters $\theta$ and $\gamma_k$ ($ 1\le k \le n$) in the POVM elements. These parameters 
are left to be determined in later proofs.

With this strategy, we have the following theorem.
\begin{theorem}
Two sequential violations of the LC bound are possible with the quantum switch.
\end{theorem}
\begin{proof}
Denote the quantum state that Charlie and Bob $k$ receive as $\rho_{CB}^k$. 
$\rho_{CB}^1$ is the initial state that Charlie and Bob 1 receive. By the measurement strategy in Eq.~\eqref{eq:measurement}
and the update rule of the postmeasurement state [Eq.~\eqref{eq:poststate}], we obtain the relation 
between $\rho_{CB}^k$  and  $\rho_{CB}^{k-1}$ as
\begin{equation}
\begin{aligned} 
\rho_{CB}^k = & \frac{2+\sqrt{1-\gamma_1^2}}{4} \rho_{CB}^{k-1}  + \frac{1}{4}  ( \mathbbm{I} \otimes \sigma_z)  \rho_{CB}^{k-1}  ( \mathbbm{I} \otimes \sigma_z) \\
&+\frac{1-\sqrt{1-\gamma_1^2}}{4}   ( \mathbbm{I} \otimes \sigma_x)  \rho_{CB}^{k-1}  ( \mathbbm{I} \otimes \sigma_x).
\end{aligned}
\end{equation}

Let us examine the shared quantum states of the first two rounds more carefully. 
In the first round, the shared state $\rho_{CB}^1$ between Charlie and Bob is 
$\ket{00}+\ket{11}$,
where the normalizing factor is omitted for simplicity. 
In the second round, the shared state between Charlie and Bob becomes
\begin{equation}
\begin{aligned} 
\{ &(\frac{2+\sqrt{1-\gamma_1^2}}{4}, \ket{00}+\ket{11}),  (\frac{1}{4}, \ket{00}-\ket{11}), \\
 &(\frac{1-\sqrt{1-\gamma_1^2}}{4}, \ket{01}+\ket{10})  \},
\end{aligned}
\end{equation}
where the normalizing factor is omitted for simplicity. The first real number in each bracket is the probability and 
the second quantity in each bracket is the quantum state with this probability.

Let us now examine the quantum value of the LC inequality of the first two rounds.
The $k$-th round LC inequality $I_Q^k$ consists of two parts, which we denote by $X^k$ and $Y^k$ 
\begin{equation}
\begin{aligned}
 X^k = &  p( b_k=0, a_2=x_1 | y_k=0 ) \\
  & + p(  b_k=1 , a_1 = x_2  | y_k=0 ), \\
 Y^k  =  & p( b_k \oplus c  = y_k z | x_1 = x_2 = 0 ).
\end{aligned}
\end{equation}
We will first examine the value of $X^k$, and then examine the value of $Y^k$ for the two rounds $k=1,2$.

Let us first examine $X^1$. For the first round, the quantum state is
 $\ket{00}+\ket{11}$.  Conditioned on $y_1=0$, Bob 1 performs a $Z$ measurement on his share 
 of the quantum state. When he obtains $b_1=0$, the control qubit becomes $\ket{0}_c$  which 
 means that $A_1$ is before $A_2$ in the time order. Hence, since both $A_1$ and 
 $A_2$ are measure-and-prepare processes, we have $ a_2 = x_1$. Likewise, when Bob 1 
 obtains  $b_1=1$, the control qubit becomes $\ket{1}_c$  which 
 means that $A_2$ is before $A_1$ in the time order. Hence, $ a_1 = x_2$. Since the probabilities 
 that Bob 1 gets $b_1=0$ and $b_1=1$ are both $1/2$, we have 
 \begin{equation}
 \begin{aligned}
 X^1   =  &  p( b_1 = 0, a_2 = x_1 | y_1 = 0 )  \\
   &+ p(  b_1=1, a_1 = x_2  | y_1=0 ) \\
 =  & p( b_1 = 0 | y_1 = 0 ) + p( b_1=1 | y_1=0 ) \\
  = & \frac{1}{2} + \frac{1}{2}  =1.
 \end{aligned}
 \end{equation}

Then let us examine $X^2$ in the second round. As mentioned previously, there are three cases
for the quantum state shared between Charlie and Bob 2, namely $\ket{00}+\ket{11}$ 
$\ket{00}-\ket{11}$, and $\ket{01}+\ket{10}$. For the first two cases,  $b_1=0$ corresponds 
to  $\ket{0}_c$ and $b_1=1$ corresponds to $\ket{1}_c$.  Therefore, the $X$ value is 
still 1 according to the previous reasoning. For the third case, $b_1=0$ corresponds 
to  $\ket{1}_c$. This means that $A_2$ is before $A_1$ on the time order. Hence $a_2$
and $x_1$ are independent. As $x_1$ is chosen uniformly randomly from $\{ 0,1\}$, the 
probability of $a_2=x_1$ is $1/2$.   For $b_1=1$ in the third case, by similar reasoning we have that
the probability of $a_1 = x_2$ is $1/2$.  As in the previous two cases, the probabilities 
 that Bob 1 gets $b_1=0$ and $b_1=1$ are both $1/2$.
Hence, the $X$ value in the third case is 
 \begin{equation}
 \begin{aligned}
 X   =  &  p( b_1 = 0, a_2 = x_1 | y_1 = 0 )  \\
   & + p(  b_1=1, a_1 = x_2  | y_1=0 ) \\
 =  & \frac{1}{2} p( b_1 = 0 | y_1 = 0 ) + \frac{1}{2} p( b_1=1 | y_1=0 ) \\
  = & \frac{1}{4} + \frac{1}{4}  = \frac{1}{2}.
 \end{aligned}
 \end{equation}
 Overall, since the first two cases have probability $(3+\sqrt{1-\gamma_1^2})/4$ and the third
 case has probability $(1-\sqrt{1-\gamma_1^2})/4$,
 we have that the overall $X$ value is 
\begin{equation}
X^2 = 1 \times \frac{3+\sqrt{1-\gamma_1^2}}{4} + \frac{1}{2} \times \frac{1-\sqrt{1-\gamma_1^2}}{4} =  1- \frac{1-\sqrt{1-\gamma_1^2}}{8 } .
\end{equation}
 
Now let us turn to $Y$. We first define the following quantity, 
\begin{equation}
 \label{eq:CHSHdef}
\begin{aligned}
I_{CHSH}^k =& 2 [ q( c = b_k | 00) +q( c = b_k |01) +q( c =    \\
    & b_k |10) +q( c \not = b_k |11) - 2 ],
\end{aligned}
\end{equation} 
which is called the $k$th-round Clauser-Horne-Shimony-Holt (CHSH) value. 
Here $q( c = b_k | 00)$ means the probability that $c = b_k$ when $z=0, y_k=0$. The other quantities are similar.
Importantly, in this definition, we assume that only two parties are involved, namely Charlie and Bob $k$. Alice 1 and 
Alice 2 do not show up in this inequality.
By the result of Ref.~\cite{brown2020arbitrarily},  
the $k$th-round CHSH value is 
\begin{equation}
I_{CHSH}^k = 2^{2-k} (\gamma_k \sin \theta + \cos \theta \prod \limits_{j=1}^{k-1} (1+\sqrt{1- \gamma_j^2}) ).
\end{equation}

We now relate $Y^k$ with $ I_{CHSH}^k$. 
First note that when $x_1 = x_2=0$, the target qubit always has the value 
of $\ket{0}$ at the end, regardless of what is the order of $A_1$ and $A_2$. 
Hence, the two probability distributions $ p ( \cdot | x_1 =0, x_2=0) $ and $q( \cdot)$ are equivalent.
 Next note that 
 \begin{equation}
 \label{eq:Yconvert}
\begin{aligned}
Y^k = & p( b_k \oplus c  = y_k z | x_1 = x_2 = 0 )  = q(b_k \oplus c  = y_k z) \\
=  & \sum \limits_{u,v\in \{ 0,1\} } q( y_k=u,z= v)  q (b_k \oplus c  = y_k z | uv )\\
  = & \sum \limits_{u,v\in \{ 0,1\} } \frac{1}{4} q (b_k \oplus c  = y_k z | uv )\\
  = &  \frac{1}{4}  (q( c = b_k | 00) +q( c = b_k |01)   \\
  & +q( c =b_k |10) +q( c \not = b_k |11)) .
\end{aligned}
 \end{equation} 
Hence by combining Eqs.~\eqref{eq:CHSHdef} and  \eqref{eq:Yconvert}, we get
\begin{equation}
I_{CHSH}^k = 2 [ 4 Y^k - 2 ].
\end{equation}
Rearranging the terms, we get
\begin{equation}
Y ^k= \frac{1}{2} +  I_{CHSH}^k / 8 .
\end{equation} 

Combining the information of $X^k$ and $Y^k$, we have that 
in the first round, the value of $X^1+Y^1$ is 
\begin{equation}
I_Q^1 =   \frac{3}{2} +  \frac{1}{4}  (\gamma_1 \sin \theta + \cos \theta  ),
\end{equation}
while  in the second round, the value of $X^2+Y^2$ is 
\begin{equation}
I_Q^2=  1- \frac{1-\sqrt{1-\gamma_1^2}}{8 } + \frac{1}{2} + \frac{1}{8}  (\gamma_2 \sin \theta + \cos \theta  (1+\sqrt{1- \gamma_1^2}) ). 
\end{equation}
We now examine whether both of these quantities can violate 
the classical bound $7/4$. Indeed, by letting $\theta = 10^{-3} $, $\gamma_1 = 10^{-3} $, and $\gamma_2 =1 $,
we have $I_Q^1  > 7/4$ and $I_Q^2 > 7/4$. 
\end{proof}

By optimizing the parameters, we have the following more quantitative result.
\begin{theorem}
With the quantum switch, the minimum quantum value of the two rounds $\min (I_{Q}^1, I_{Q}^2)$ is at least 1.7640,
 larger than the classical upper bound $1.75$.
\end{theorem}
\begin{proof}
We optimize $\theta $ and $\gamma_1$ such that $ \min (I_Q^1, I_Q^2) $ is maximized.
First we assume $\theta = \gamma_1$ and let $I_Q^1= I_Q^2$, we obtain that in this case
$\theta = \gamma_1=0.3312 $, $I_Q^1= I_Q^2 = 1.7633 > 7/4$.
Next, we explore the case $\theta \not = \gamma_1$ and maximize $ \min (I_Q^1, I_Q^2) $ around $(\theta, \gamma_1)=(0.3312,0.3312)$; we obtain 
$\theta = 0.411$, $\gamma_1=0.349 $, $I_Q^1= 1.7640$, and $I_Q^2 = 1.7641$.
\end{proof}

We then compare our scheme with the scheme proposed in Ref.~\cite{van2022device}. Note that the 
scheme presented in Ref.~\cite{van2022device} is a special case of our scheme, with parameters $\theta=\pi/4$ and $\gamma_1=\gamma_2=1$.
The comparison is summarized in Table~\ref{tab:compare}. In the first round, our scheme exhibits a smaller violation value compared to the scheme described in Ref.~\cite{van2022device}. However, in the second round, our scheme achieves a larger violation value in contrast to the scheme of Ref.~\cite{van2022device}. Notably, the violation of the scheme proposed in Ref.~\cite{van2022device} completely vanishes during the second round, resulting in an infinite ratio of violation when comparing our scheme to theirs.

\begin{table}[htb] 
\caption{Comparison of the violation amount $I_Q-I_{LC}$.}
\centering
\begin{tabular}{ccc}
\hline
   & Round 1 & Round 2 \\ 
\hline
Ref.~\cite{van2022device} &  \textbf{0.1036} & 0 \\
This work & 0.0140 & \textbf{0.0140}  \\
\hline
\end{tabular}
\label{tab:compare}
\end{table}

\section{Detection loophole}In reality, the detector in the device-independent test is not perfect. In particular, 
photonic detectors do not have unit efficiency, which may induce detection loopholes in the test \cite{massar2002nonlocality, vertesi2010closing, christensen2013detection, cao2016tight, liu2018high, cao2021detection}. 
Therefore, in this section, we examine what is the detection efficiency requirement of the detectors
for a successful sequential DI test of the indefinite causal order. 

Let us start by describing the detector model, as illustrated in Fig.~\ref{fig:detector}(a). A detector with efficiency $\eta$ has a probability of $1-\eta$ of detecting nothing and a probability of $\eta$ of functioning as an ideal detector. It is assumed that all detectors used in the experiment have the same efficiency, although our method can be extended to detectors with varying efficiencies. In the latter case, we will obtain a range of permissible efficiencies for each detector, rather than a single figure of merit.

\begin{figure}[htb]
\centering \includegraphics[width=8.5cm]{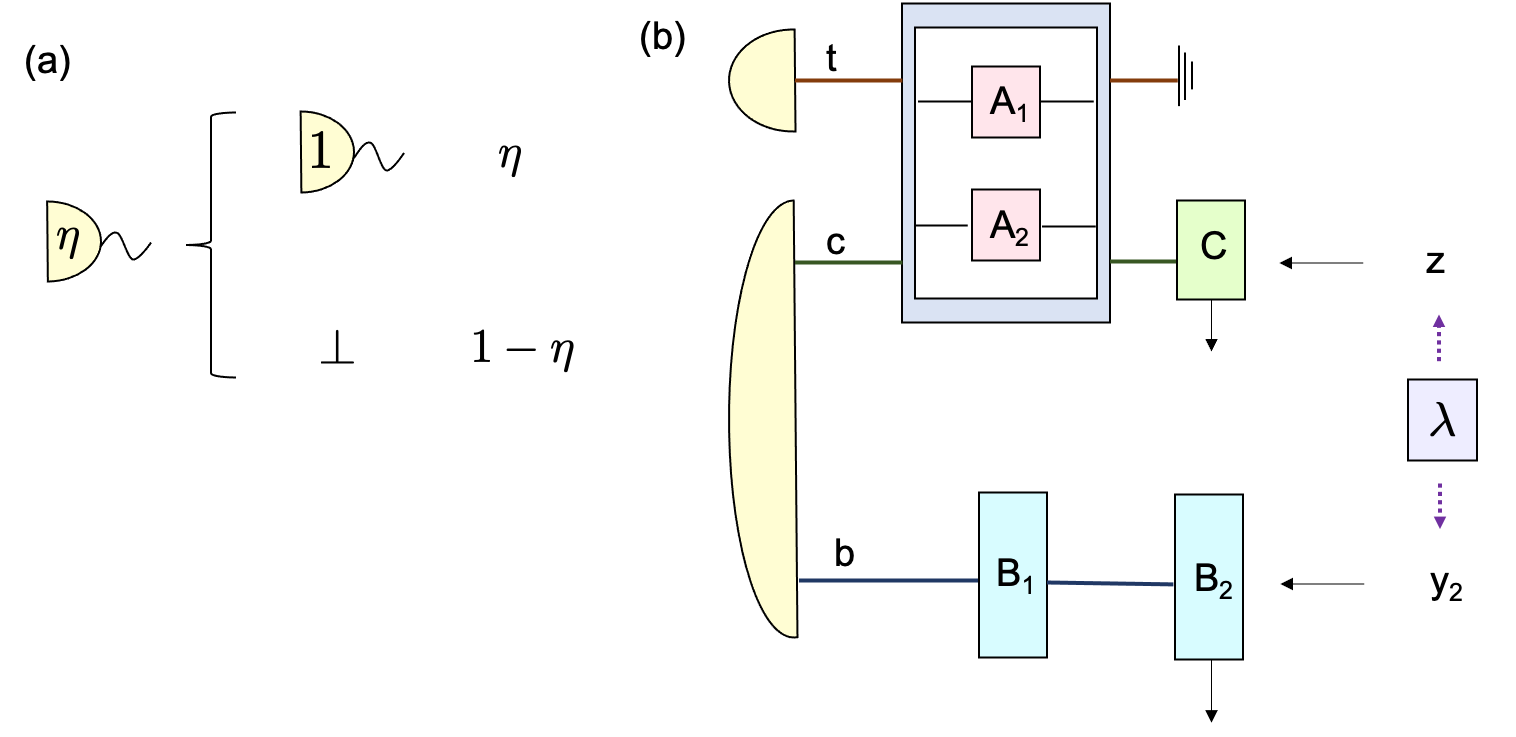}
\caption{(a) The detector model. (b) The model of the randomness loophole.}
\label{fig:detector}
\end{figure}

The result of this section is summarized in the following theorem.
\begin{theorem}
When the detection efficiency $\eta$ is at least $99.6\%$, it is possible to sequentially device-independently certify the quantum 
switch two times.
\end{theorem}
\begin{proof}
Let $p^\eta$ denote the probability distribution where all detectors have efficiency $\eta$. Let us first
examine the quantity $ p^\eta  ( b=0, a_2=x_1 | y=0 )$.  For this quantity, there are two detectors involved, 
i.e., Bob and Alice 2's detectors. Since both detectors have efficiency $\eta$, we have
\begin{equation}
 p^\eta  ( b=0, a_2=x_1 | y=0 ) = \eta^2  p( b=0, a_2=x_1 | y=0 ).
\end{equation}
By similar reasoning, we have $ p^\eta  (  b=1 , a_1 = x_2  | y=0 ) = \eta^2  p(  b=1 , a_1 = x_2  | y=0 )$
and $p^\eta  ( b \oplus c  = yz ) ) = \eta^2  p( b \oplus c  = yz ) )$.

Let $I_Q ^ \eta$ be the L-C quantity when the detectors have efficiency $\eta$. Then we have
\begin{equation}
\begin{aligned} 
I_Q ^ \eta 
& =  p^\eta( b=0, a_2=x_1 | y=0 ) + p^\eta(  b=1  \\
&\quad\quad  , a_1 = x_2  | y=0 ) + p^\eta( b \oplus c  = yz ) \\
& = \eta^2( p( b=0, a_2=x_1 | y=0 ) + p(  b=1  \\
&\quad\quad , a_1 = x_2  | y=0 ) + p( b \oplus c  = yz ) ) \\
& =  \eta^2 I_Q  > \frac{7}{4}.
\end{aligned} 
\end{equation}
In Theorem 3, we have that the maximum $I_Q$ for two rounds is 1.7640; hence the minimum detection efficiency
for violating the LC inequality two times is
\begin{equation}
\eta = \sqrt{1.75/ 1.7640} = 99.6\%.
\end{equation}
\end{proof}

\section{Randomness loophole}There is a free-will assumption in the derivation of the LC bound, which may be violated in practice \cite{yuan2015randomness}. In this section we examine to what extent 
this assumption can be relaxed. Consider the model illustrated in Fig.~\ref{fig:detector}(b).
Here, instead of $z$ and $y_k$ being randomly chosen, there is a hidden variable $\lambda$ that controls the probability 
distribution of $z$ and $y_k$ by $q( z,y_k|\lambda)$. Below, we always consider the $k$th round, and
write $y_k$ as $y$ for notation simplicity.

Let $P  = \max_{z,y,\lambda}  q(z,y | \lambda)$.
We also assume that Charlie and Bob's inputs are uncorrelated, namely,
\begin{equation}
q(z,y| \lambda) = q(z| \lambda) q(y| \lambda).
\end{equation}
Let us digest the meaning of $P$. In one extreme, $P=1$ corresponds to that there
is absolutely no free will for a certain $\lambda$. For that $\lambda$, the values 
of $z$ and $y$ are deterministic. In the other extreme,  $P=1/4$ corresponds to perfect free will, because it implies that
for any $z,y,\lambda$,  $q(z,y | \lambda)= 1/4$. Between the extremes, the smaller $P$ is, the more
free will we have.
Therefore, we only need to find the maximum $P$ such that the indefinite causal order can still be certified.
The following theorem summarizes the main result of this section.
\begin{theorem}
The minimum free will required is $P \le 0.2640$ for two successful sequential DI tests of the indefinite causal order.
\end{theorem}
\begin{proof}
The quantity $I_{LC}^k$ consists of two parts,
\begin{equation}
I_{LC}^k = \alpha^k+ Y^k, 
\end{equation}
where 
\begin{equation}
\begin{aligned}
 \alpha^k = &  p( b_k=0, a_2=x_1 | y_k=0 ) 
 \\ &
   + p(  b_k=1 , a_1 = x_2  | y_k=0 ), \\
 Y^k  =  & p( b_k \oplus c  = y_k z | x_1 = x_2 = 0 ).
 \end{aligned}
\end{equation}
Next we examine the effect of no free will on $\alpha^k$ and $Y^k$ 
separately. 

We first start from $\alpha^k$. According to the definition of $\mathcal{LC}$, it is a convex combination of $\mathcal{LC}_1$ and $\mathcal{LC}_2$,
where in $\mathcal{LC}_1$, $A_1$ is causally ordered before $A_2$ and in $\mathcal{LC}_2$, $A_2$ is causally ordered before $A_1$.
Hereafter, we prove the case for $\mathcal{LC}_1$. The proof for  $\mathcal{LC}_2$ is similar. Then by a convex combination, we finish
the proof for $\mathcal{LC}$.

For $\mathcal{LC}_1$, we have the following equation:
\begin{equation}
p( b_k=0, a_2=x_1 | y_k=0 ) \le p( b_k=0 | y_k=0 ).
\end{equation}
Moreover, since $A_1$ is causally before $A_2$ and we have assumed $A_2$ has perfect free will to choose her setting $x_2$ (only Bob $k$ and Charlie
are restricted in their free will), we have
\begin{equation}
p( b_k=1, a_1=x_2 | y_k=0 ) \le \frac{1}{2} p( b_k=1 | y_k=0 ).
\end{equation}
Combining these two equations, we obtain 
\begin{equation}
p( b_k=0 | y_k=0 ) \ge 2\alpha^k - 1.
\end{equation}
Combined with the fact
\begin{equation}
p( b_k=0 | y_k=0 )  \le  1,
\end{equation}
we have
\begin{equation}
\alpha^k  \le 1.
\end{equation}

Let us then examine  $Y^k$. As previously, we transform $Y^k$ to $I_{CHSH}^k$, which
are related by  
\begin{equation}
Y^k = \frac{1}{2} +  I_{CHSH}^k / 8 .
\end{equation} 
Here, $I_{CHSH}^k$ has the expression
\begin{equation}
\begin{aligned}
I_{CHSH}^k =& 2 [ q( c = b_k | 00) +q( c = b_k |01) +q( c =    \\
    & b_k | 10) +q( c \not = b_k |11) - 2 ].
\end{aligned}
\end{equation} 
Let us consider the classical local strategy that $c$ and $b_k$ are always assigned zero,
 regardless of the values of $z$ and $y_k$, and show the upper bound of $I_{CHSH}^k$.
 Other classical strategies can be similarly analyzed, and the upper bound remains the same. 
With this strategy, we can write $I_{CHSH}^k$ as
\begin{equation}
\begin{aligned}
I_{CHSH}^k  =  &4 [q_C(0 |\lambda)  q_B(0 |\lambda) + q_C(0 |\lambda) q_B(1 |\lambda)  \\
& + q_C (1 |\lambda) q_B(0 |\lambda)  - q_C (1 |\lambda) q_B(1 |\lambda) ].
\end{aligned}
\end{equation}

For further analysis, we define two additional symbols $P_C$ and $P_B$, which satisfy
\begin{equation}
P_C = \max_{x, \lambda} \{ q_C (x |\lambda ) \},   P_B = \max_{x, \lambda} \{ q_C (x |\lambda ) \}.
\end{equation}
These two quantities are related with $P$ by 
\begin{equation}
P = P_C P_B.
\end{equation}
By the definition of $P_C$ and $P_B$, we have
\begin{equation}
q_C(1 |\lambda) = 1- q_C(0 |\lambda) \ge 1- P_C,
\end{equation}
and 
\begin{equation}
q_B(1 |\lambda) = 1- q_B(0|\lambda)  \ge 1- P_B.
\end{equation}
Therefore,
\begin{equation}
\begin{aligned}
 &  q_C(0 |\lambda)  q_B(0 |\lambda) + q_C(0 |\lambda) q_B(1 |\lambda) \\
   & \quad  + q_C (1 |\lambda) q_B(0 |\lambda)  -q_C (1 |\lambda) q_B(1 |\lambda)  \\
 =& 1- 2 q_C (1 |\lambda) q_B(1 |\lambda) \\
 \le & 1- 2(1-P_C )  (1-P_B).
 \end{aligned}
\end{equation}
Hence
\begin{equation}
\label{eq:chshfree}
I_{CHSH}^k \le 4 [1- 2(1-P_C)(1-P_B)] = 8( P_C + P_B - P) -4.
\end{equation}
The maximum of the right-hand side of Eq.~\eqref{eq:chshfree} is achieved when 
$  P_B=1/2$ and $ P_C = 2P$.  The maximum value achieved is
\begin{equation}
I_{CHSH}^k = 8P.
\end{equation}

By combining the value of $\alpha^k$ and $Y^k$, we have 
\begin{equation}
I_{LC}^k \le 1 + \frac{1}{2} + \frac{8P}{ 8}.
\end{equation}
When the classical bound $I_{LC}^k$ reaches  $I_Q^k = 1.7640$, 
no quantum violation will be possible. Hence, by solving
\begin{equation}
1 + \frac{1}{2} + \frac{8P}{ 8} = 1.7640,
\end{equation}
we obtain the maximum allowable $P$ as
\begin{equation}
P_{\max} = 0.2640.
\end{equation}
This finishes the proof of the theorem.
\end{proof}
Equivalently, the minimum min-entropy of the inputs $z$ and $y$ should be at least $\log_2(1/0.2640)=1.92$ for two successful sequential DI tests.

\section{Arbitrary large number of sequential tests}In this section, we examine the possibility of an arbitrary number of sequential DI tests of the indefinite causal order. 
Essentially, we need to check whether for all $1 \le k \le n$, the term $I_{Q}^k$ is larger than $7/4$.
To do so, we first need to get a lower bound on $I_{Q}^k$. 
To this end, we have the following lemma.
\begin{lemma}
 In the $k$th round, the value of $I_{Q}^k$ satisfies 
\begin{equation}
\begin{aligned}
I_{Q}^k \ge  & 1 + \frac{1}{2} \left( \prod\limits_{j=1}^{k-1} \left( \frac{3+\sqrt{ 1-\gamma_{j}^2}}{4}\right)  -1 \right)  + \frac{1}{2}  \\
& + \frac{1}{8} 2^{2-k} \left(\gamma_k \sin \theta + \cos \theta  \prod\limits_{j=1}^{k-1}  (1+\sqrt{1- \gamma_j^2}) \right).
\end{aligned}
\end{equation}
\end{lemma}
\begin{proof}
As $I_{Q}^k$ can be decomposed 
into two parts, $X^k$ and $Y^k$, we examine these two parts one by one in the following.

First, let us examine $X^k$.
Note first that the postmeasurement state of the $k$th round is related with 
the postmeasurement state of the ($k-1$)th round by 
\begin{equation}
\begin{aligned} 
\rho_{CB}^k = & \frac{2+\sqrt{1-\gamma_1^2}}{4} \rho_{CB}^{k-1}  + \frac{1}{4}  ( \mathbbm{I} \otimes \sigma_z)  \rho_{CB}^{k-1}  ( \mathbbm{I} \otimes \sigma_z) \\
&+\frac{1-\sqrt{1-\gamma_1^2}}{4}   ( \mathbbm{I} \otimes \sigma_x)  \rho_{CB}^{k-1}  ( \mathbbm{I} \otimes \sigma_x).
\end{aligned}
\end{equation}
Let us understand the three terms on the right-hand side of the above equation in more detail. Let us consider the case that
the initial state is $\ket{\Phi^+}=\ket{00} + \ket{11}$ or $\ket{\Phi^-}=\ket{00} - \ket{11}$.
For the first term, it keeps the state unchanged; i.e.,  $\ket{\Phi^+}$ remains $\ket{\Phi^+}$ and $\ket{\Phi^-}$ remains $\ket{\Phi^-}$. For the second term, it applies a $Z$ gate on the second qubit of the quantum state. This, in particular, implies that $\ket{\Phi^+}$ and $\ket{\Phi^-}$ are interchanged. For the third term,  it applies an $X$ 
gate on the second qubit of the quantum state. Then $\ket{00} + \ket{11}$ becomes $\ket{\Psi^+} = \ket{01} + \ket{10}$ and $ \ket{00} - \ket{11}$ becomes $\ket{\Psi^-} = \ket{01} - \ket{10}$, which goes out of the space spanned by $\{\ket{\Phi^+}, \ket{\Phi^-}  \}$.  

As analyzed previously, when the shared quantum state between Charlie and Bob $k$ is $\ket{\Phi^+}$ or $\ket{\Phi^-}$, $X^k$ has value 1, and when the shared quantum state between Charlie and Bob $k$ is $\ket{\Psi^+}$ or $\ket{\Psi^-}$, $X^k$ has value $1/2$. 
We now estimate the probability that the quantum state shared between Charlie and Bob $k$ is $\ket{\Phi^+}$ or $\ket{\Phi^-}$. 
According to the previous analysis, when the initial state is $\ket{\Phi^+}$ or $\ket{\Phi^-}$, for the first two cases, after the update, the state remains 
$\ket{\Phi^+}$ or $\ket{\Phi^-}$. The total probability of the first two cases for the $k$th round is $(3+\sqrt{1-\gamma_{k-1}^2})/4$. 
The probability that all $k$ rounds stay in the first two cases is at least $\prod_{j=1}^{k-1} (3+\sqrt{1-\gamma_{j}^2})/4$.
Therefore, we have a bound on $X^k$ as 
\begin{equation}
\begin{aligned}
X ^k  & \ge  \prod\limits_{j=1}^{k-1} \left( \frac{3+\sqrt{ 1-\gamma_{j}^2}}{4}\right)  \times  1  + ( 1- \\
 & \quad \quad \quad  \prod\limits_{j=1}^{k-1} \left( \frac{3+\sqrt{ 1-\gamma_{j}^2}}{4}\right) ) \times \frac{1}{2} \\
& = 1 + \frac{1}{2} \left( \prod\limits_{j=1}^{k-1} \left( \frac{3+\sqrt{ 1-\gamma_{j}^2}}{4}\right)  -1 \right) .
\end{aligned}
\end{equation}

For $Y^k$, we first use a result from Ref.~\cite{brown2020arbitrarily}, which states that $I_{CHSH}^k$ satisfies
\begin{equation}
 I_{CHSH}^k= 2^{2-k} \left(\gamma_k \sin \theta + \cos \theta  \prod\limits_{j=1}^{k-1}  (1+\sqrt{1- \gamma_j^2}) \right).
\end{equation}
Together with the relation
\begin{equation}
Y^k = \frac{1}{2} + \frac{I_{CHSH}^k}{8},
\end{equation}
we obtain
\begin{equation}
Y^k = \frac{1}{2} + \frac{1}{8} 2^{2-k} \left(\gamma_k \sin \theta + \cos \theta  \prod\limits_{j=1}^{k-1}  (1+\sqrt{1- \gamma_j^2}) \right).
\end{equation}

Therefore, in the $k$th round, the value of $I_{Q}^k=X^k+Y^k$ is 
\begin{equation}
\begin{aligned}
I_{Q}^k \ge  & 1 + \frac{1}{2} \left( \prod\limits_{j=1}^{k-1} \left( \frac{3+\sqrt{ 1-\gamma_{j}^2}}{4}\right)  -1 \right)  + \frac{1}{2}  \\
& + \frac{1}{8} 2^{2-k} \left(\gamma_k \sin \theta + \cos \theta  \prod\limits_{j=1}^{k-1}  (1+\sqrt{1- \gamma_j^2}) \right).
\end{aligned}
\end{equation}
\end{proof}

Our task is then reduced to choosing $\gamma_k$ ($1\le k\le n$) and $\theta$ such that $I_{Q}^k > 7/4$ for all $1\le k \le n$ and any arbitrary large $n$. 
The condition $I_{Q}^k > 7/4$ is equivalent to 
\begin{equation}
\label{eq:gamma_cond}
\begin{aligned}
\gamma_k > & \frac{1}{\sin \theta} [  2^{k-1} - \cos \theta  \prod\limits_{j=1}^{k-1}  (1+\sqrt{1- \gamma_j^2})  \\
 &+ 2^k (1-  \prod\limits_{j=1}^{k-1} \left( \frac{3+\sqrt{ 1-\gamma_{j}^2}}{4}\right) )   ].
\end{aligned}
\end{equation}
Therefore, for constructing valid $\gamma_k$, we construct 
 the following sequence
\begin{equation}
\label{eq:gammatheta} 
\begin{aligned}
\gamma_k(\theta) = 
&
(1+\epsilon) \frac{1}{\sin \theta} [  2^{k-1} - \cos \theta  \prod\limits_{j=1}^{k-1}  (1+\sqrt{1- \gamma_j^2})    \\
& + 2^k (1-  \prod\limits_{j=1}^{k-1} \left( \frac{3+\sqrt{ 1-\gamma_{j}^2}}{4}\right) )   ] .
\end{aligned}
\end{equation}
For positive $\epsilon>0$, we clearly have $\gamma_k(\theta)$ satisfy Eq.~\eqref{eq:gamma_cond}.

Equation~\eqref{eq:gammatheta} only makes sense if for any $j<k$,  we have $\gamma_j <1$. Therefore,
before proving our main theorem, we first prove the following lemma.
\begin{lemma}
\label{lemma:critical}
For any $n$, there exists a $\theta \in (0, \pi/4)$ such that $ \gamma_k(\theta) <1 $ for all $1\le k \le n$.
\end{lemma}
\begin{proof}
By substituting the inequalities $\forall x\in [0,1],  \sqrt{1- x^2 }  \ge  1 - x^2$,  $\forall \theta \in (0,\pi/4],  \cos \theta  \ge  1 - \theta^2/2$,
and $\forall \theta \in (0,\pi/4],   \sin \theta  \ge  \theta/2$, we get 
\begin{equation}
\label{eq:appproof}
\begin{aligned}
\gamma_k(\theta)  \le 2^k  (1+\epsilon) \frac{ \theta^2/2 + \sum\limits_{j=1}^{k-1} \gamma_j^2/2 + 2  \sum\limits_{j=1}^{k-1} \gamma_{j}^2/4   }{\theta },
\end{aligned} 
\end{equation}
the derivation of which can be found in Appendix~\ref{appsec:eqder}.

Define 
\begin{equation}
\label{eq:p}
p_k(\theta) = 
\begin{cases}
2^k  (1+\epsilon) \frac{ \theta^2/2 + \sum\limits_{j=1}^{k-1} p_j^2/2 + 2  \sum\limits_{j=1}^{k-1} p_{j}^2/4   }{\theta } & \textrm{if } p_{j} \in (0,1),  \forall j<k; \\
\infty & \textrm{otherwise}. \\
\end{cases}
\end{equation}

Since the right-hand side of Eq.~\eqref{eq:p} increases with $p_j$ for any $j<k$, then $p_k \ge \gamma_k$. 

We now prove the statement that there exists $\theta_k \in (0, \pi/4]$ such that for all $1 \le j \le k$ and any $\theta \in (0, \theta_{k})$, $p_j( \theta) < 1$.

By direct calculation, we obtain 
\begin{equation}
p_1(\theta) = (1+\epsilon ) \theta.
\end{equation}
By taking $\theta_1 = 1/(2+\epsilon)$, the statement holds for $k=1$.

We now prove by induction that 
\begin{equation}
p_k(\theta) = C_{k, \epsilon} \theta
\end{equation}
holds for any $k$, where $ C_{k, \epsilon}$ is a positive constant that depends only on $k$ and $\epsilon$ but not $\theta$.
Clearly, the statement holds for $k=1$, by choosing $C_{1, \epsilon} =1+\epsilon $.  
Now assuming the statement holds for all $j \le k-1$, we now consider the case $k$.
By Eq.~\eqref{eq:p}, we have
\begin{equation}
\label{eq:p1}
\begin{aligned}
p_k(\theta)  & = 2^k  (1+\epsilon) \frac{ \theta^2/2 + \sum\limits_{j=1}^{k-1} p_j^2/2 + 2  \sum\limits_{j=1}^{k-1} p_{j}^2/4   }{\theta }  \\
& =  2^k  (1+\epsilon) \frac{ \theta^2/2 + \sum\limits_{j=1}^{k-1} C_{j, \epsilon}^2\theta^2/2 + 2  \sum\limits_{j=1}^{k-1} C_{j, \epsilon}^2\theta^2/4   }{\theta }  \\
& =  2^k  (1+\epsilon) (\theta/2 + \sum\limits_{j=1}^{k-1} C_{j, \epsilon}^2\theta/2 + 2  \sum\limits_{j=1}^{k-1} C_{j, \epsilon}^2\theta/4  ). \\
\end{aligned}
\end{equation}
By letting 
\begin{equation}
C_{k, \epsilon} = 2^k  (1+\epsilon) (1/2 + \sum\limits_{j=1}^{k-1} C_{j, \epsilon}^2/2 + 2  \sum\limits_{j=1}^{k-1} C_{j, \epsilon}^2/4  )    ,
\end{equation}
we have $p_k(\theta) = C_{k, \epsilon} \theta$.

Now choosing $\theta_k = \max \{ \pi/4,   1/( 2\max_{1\le j \le k } C_{j, \epsilon} ) \}$, we have  $p_j(\theta)< 1$ for all $1\le j\le k$ and all $\theta \in (0, \theta_{k})$.  Combined with the fact that $p_j \ge \gamma_j$ for any $1\le j \le k$, we have 
that $\gamma_j < 1$ for all $1 \le j \le k$, which finishes the proof.  
\end{proof}

Now we are ready to prove  the main theorem.
\begin{theorem}
For any $n$, there exists a suitable choice of $\theta$ and $0< \gamma_k \le 1$ ($1 \le k \le n$) such 
that  $I_Q^k  > 7/4 $ for all $ 1\le k\le n$.
\end{theorem}

\begin{proof}
Choose $\theta$ according to Lemma~\ref{lemma:critical} and let $\gamma_k =  \gamma_k(\theta)$. 
By the definition of $ \gamma_k(\theta)$, we have $I_Q^k  > 7/4 $ for any $1\le k \le n$ and this 
concludes the proof.
\end{proof}

\begin{figure}[htb]
\centering \includegraphics[trim={3cm 0 3cm 1cm},clip,width=8.5cm]{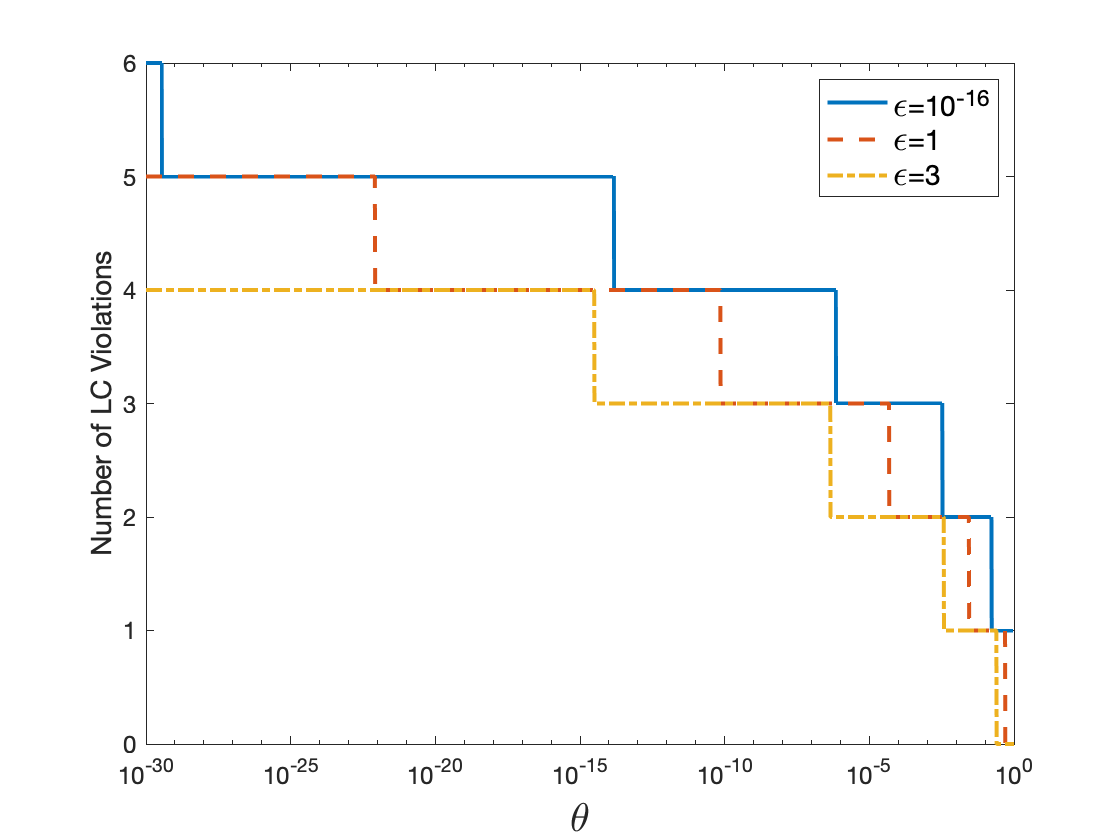}
\caption{The number of LC violations as a function of $\theta$ for different values of $\epsilon$.}
\label{fig:thetaVSk}
\end{figure}

We next explore the relation between $\theta$ and the maximum number of allowable violations $n$. By numerical
simulations, the result is shown in Fig.~\ref{fig:thetaVSk}. It can be seen that $\theta$ decreases more than exponentially
with respect to $n$ for any choice of the positive constant $\epsilon$. Therefore, the margin of violation also 
decreases more than exponentially with respect to $n$.

\section{Discussion}In this work, we studied sequential DI certification of the indefinite causal order in a quantum switch. We 
examined the case of two sequential violations and showed its possibility. We subsequently 
gave the maximum violation value for two sequential violations. We then moved
on to examine potential loopholes in the test. We showed that the detection loophole can be
avoided as long as the efficiency of the detectors is at least $99.6\%$. We also showed that
the randomness loophole can be avoided when the min-entropy of the inputs is at least 1.92. 
Finally, we showed that an arbitrary number of sequential certifications of an indefinite causal order is also
possible by constructing an explicit strategy.

There are a few interesting future directions. First, it is interesting to explore sequential
DI tests of other indefinite causal order phenomena, such as the one in \cite{oreshkov2012quantum}.
Second, in this work, we examined the case that one side of the spacelike separated parties 
is sequentially tested to obtain repeated violations. It is also interesting to explore the case where
both sides of the spacelike separated parties are sequentially tested to obtain repeated violations.
Third, since the margin of two sequential violations of the local-causal inequality is substantial, an
experimental demonstration of this sequential DI test is worth pursuing.

\begin{acknowledgements}
This work was supported by the National Natural Science Foundation of China (Basic Science Center Program 61988101), the National Natural Science Foundation of China (12105105), the Natural Science Foundation of Shanghai (21ZR1415800), the Shanghai Sailing Program (21YF1409800), the startup fund from East China University of Science and Technology (YH0142214) and the Shanghai AI Lab.
\end{acknowledgements}

\appendix 

\section{Derivation of Equation~\eqref{eq:appproof}}
\label{appsec:eqder}

Following is the derivation of Eq.~\eqref{eq:appproof} from the proof of Lemma \ref{lemma:critical}:
\begin{widetext}
\begin{equation}
\begin{aligned}
\gamma_k(\theta) & \le 
(1+\epsilon)  \frac{ 2^{k-1} - (1 - \theta^2/2)  \prod\limits_{j=1}^{k-1}  (2- \gamma_j^2) + 2^k (1-  \prod\limits_{j=1}^{k-1} \left( \frac{4 -\gamma_{j}^2}{4}\right) )    }{\theta/2}  \\
& = 2^k  (1+\epsilon) \frac{1-(1 - \theta^2/2)  \prod\limits_{j=1}^{k-1}  (1- \gamma_j^2/2) + 2- 2 \prod\limits_{j=1}^{k-1} ( 1- \gamma_{j}^2/4 ) }{\theta} \\
& \le 2^k  (1+\epsilon) \frac{ \theta^2/2 + \sum\limits_{j=1}^{k-1} \gamma_j^2/2 + 2  \sum\limits_{j=1}^{k-1} \gamma_{j}^2/4   }{\theta }.
\end{aligned} 
\end{equation}
\end{widetext}



%

\end{document}